\documentclass[a4paper,UKenglish,cleveref, autoref, thm-restate]{lipics-v2019}

\title{A Constant-Factor Approximation Algorithm for Point Guarding an Art Gallery} 

\titlerunning{Constant-approximation algorithm for the point guarding problem} 

\author{Arash Vaezi}{Sharif University of Technology, Tehran, Iran}{avaezi@ce.sharif.edu}{}{}

\authorrunning{A.\,Vaezi} 

\Copyright{Arash Vaezi} 



\category{} 

\relatedversion{} 

\supplement{}

\acknowledgements{We want to thank Dr. M. Salavatipour for introducing the Quasi-Uniform sampling technique, and Dr. Edouard Bonnet for introducing us a grid that hands over a dense decomposition for the given simple polygon.}

\nolinenumbers 

\EventEditors{John Q. Open and Joan R. Access}
\EventNoEds{2}
\EventLongTitle{42nd Conference on Very Important Topics (CVIT 2016)}
\EventShortTitle{CVIT 2016}
\EventAcronym{CVIT}
\EventYear{2016}
\EventDate{December 24--27, 2016}
\EventLocation{Little Whinging, United Kingdom}
\EventLogo{}
\SeriesVolume{42}
\ArticleNo{23}

\usepackage{graphicx}
\usepackage{hyperref}
\usepackage{comment}
\usepackage{algorithm}
\usepackage[noend]{algpseudocode}
\usepackage{amssymb}
\usepackage{amsmath}
\usepackage{wrapfig}

\newtheorem{pro}{Problem}

\def\for{\mbox{\it for}}
\def\rf{\mbox{\it rf}}
\def\VP{\mbox{\it VP}}
\def\Vp{\mbox{\it Vp}}
\def\VL{\mbox{\it VL}}

\def\NP{\mbox{\it NP}}
\def\WVP{\mbox{\it WVP}}

\def\CVP{\mbox{\it CVP}}
\def\OPT{\mbox{\it OPT}}
\def\SCR{\mbox{\it SCR}}

\def\Findtsr{\mbox{\it Findtsr}}
\def\inPL{\mbox{\it inPL}}
\def\seg#1{\overline{#1}}
\def\P{\cal P}

\def\S{\cal S}
\def\F{\cal F}
\def\M{\cal M}
\def\G{\cal G}
\def\D{\cal D}
\begin{document}

\maketitle

\begin{abstract}
Given a simple polygon $\P$, in the Art Gallery problem the goal is to find the minimum number of guards needed to cover the entire $\P$, where a guard is a point and can see another point $q$ when $\overline{pq}$ does not cross the edges of $\P$.
 This paper studies a variant of the Art Gallery problem in which guards are restricted to lie on a dense grid inside $\P$.
In the general problem, guards can be anywhere inside or on the boundary of $\P$. The general problem is called the \emph{point} guarding problem. It was proved that the point guarding problem is APX-complete, meaning that we cannot do better than a constant-factor approximation algorithm unless $P = \NP$. 
 A huge amount of research is committed to the studies of combinatorial and algorithmic aspects of this problem, and as of this time, we could not find a constant factor approximation for simple polygons. The last best-known approximation factor for the point guarding a simple polygon was $\mathcal{O}(\log (\OPT))$ introduced by É. Bonnet and T. Miltzow in 2020.
Here, we propose an algorithm with a constant approximation factor for the point guarding problem where the location of guards is restricted to a grid. The running time of the proposed algorithm depends on the number of cells of the grid. The approximation factor is constant regardless of the grid we use, the running time could be super-polynomial if the grid size becomes exponential. 
 
\end{abstract}
\section{Introduction and Related Works}
Consider a simple polygon $\P$ with $n$ vertices.   
The maximal sub-polygon of $\P$ visible to a point $q$ in $\P$ is called the \emph{visibility polygon} of $q$, which is denoted by $\VP(q)$. 
There are linear-time algorithms to compute $\VP(q)$ when the viewer is a point \cite{1}.
 For the segment $\seg{pq}$ inside $\P$, the \emph{weak visibility polygon} of $\seg{pq}$, denoted as $\WVP(\seg{pq})$, is the maximal sub-polygon of $\P$ visible to at least one point (not the endpoints) of $\seg{pq}$.
A polygon $\cal Q$ inside $\P$ is said to be \emph{completely visible} from $\seg{pq}$ if for every point $z \in \cal Q$ and for any point $w \in \seg{pq}$, $w$ and $z$ are visible (denoted as $\mathit{CVP}$ short from completely visible polygon). Also, $\cal Q$ is said to be \emph{strongly visible} from $\seg{pq}$ if there exists a point $w \in \seg{pq}$ such that for every point $z \in \cal Q$, $w$ and $z$ are visible ($\mathit{SVP}$). These different visibility types can be computed in linear time (\cite{2,avis}).

 In computational geometry, Art Gallery problems are motivated by the question, “How many security cameras are required to guard an art gallery?” The art gallery is modeled as a connected polygon $\P$. A camera, which we will henceforth call a guard, is modeled as a point in the polygon.
\emph{The Art Gallery problem} is to determine the minimum number of guards that are sufficient to see every point in the interior of an Art Gallery room.
The Art Gallery can be viewed as a polygon $\P$ of $n$ vertices, and the guards are stationary points in $\P$. Steve Fisk \cite{cite:fisk} proved that $\lfloor \frac{n}{3} \rfloor$ guards are always sufficient and sometimes necessary. Five years earlier, Victor Klee \cite{cite:victorklee} has posed this question to Václav Chvátal, who soon gave a more complicated solution. 

If guards are placed at vertices of $\P$, they are called \emph{vertex guards}~\cite{ComGeoBook}.
If guards are placed at any point of $\P$, they are called \emph{point guards}. If guards are allowed to be placed along the boundary of $\P$, they are called \emph{boundary-guards} (on the perimeter).
Since guards placed at points or vertices are stationary, they are referred to as \emph{stationary guards}. If guards are mobile along a segment inside $\P$, they are referred to as \emph{mobile guards}. If mobile guards move along edges of $\P$, they are referred to as \emph{edge guards}. 

The Art Gallery problem was proved to be NP-hard first for polygons with holes by~\cite{cite:rou&sup}. For guarding simple polygons, it was proved to be NP-complete for vertex guards by~\cite{lee}. This proof was generalized to work for point guards by~\cite{cite:aggarval}. For guarding simple polygons, the article~\cite{cite:Eidenbenz} proved that the problem is APX-complete, meaning that we cannot do better than a constant-factor approximation algorithm unless $P = \NP$. For guarding polygons with holes,~\cite{cite:Eidenbenz2} proved that if there is no restriction on the number of holes, the minimization problem
is as hard to approximate as Set Cover in general. It, therefore, follows from results about the inapproximability of Set Cover by~\cite{cite:Feige,cite:Raz&Safra} that, for polygons with holes, it is NP-hard to find a $o(\log n)$-approximation. These hardness results hold whether we are dealing with vertex guards, perimeter guards, or point guards.

Denote an optimal set of guards by $\OPT$. In 2006 Efrat and Har-Peled proposed an exact algorithm for the point guarding problem with running time at most $\mathcal{O}(n{\small|\OPT|})^{3(2{\small|\OPT|}+1)}$, where $|\OPT|$ is the size of the optimal solution~\cite{16}.
Ghosh~\cite{cite:Ghosh} provided an $\mathcal{O}(\log n)$-approximation algorithm for guarding polygons with or without holes with \emph{vertex} guards. His algorithm decomposes the input polygon into a polynomial number of cells such that each point in a given cell is seen by the same set of vertices. So, each vertex corresponds to a subset of the cells. The minimum number of vertices (subsets) that cover all of the cells cover $\P$. King and Krikpatrick obtained an approximation factor of $\mathcal{O}(\log|\OPT|)$ for
vertex guarding simple polygons and $\mathcal{O}((\log h)(\log |\OPT|))$ for vertex guarding a polygon with $h$ holes. 
They presented an $\mathcal{O}(\log \log |\OPT|)$-approximation algorithm for guarding simple polygons, using either vertex guards or perimeter guards, and the running time is polynomial in $n$ and the number of potential guard locations~\cite{cite:k}. 
In 2007 article~\cite{cite:pslogn} proposed a pseudopolynomial time $\mathcal{O} (\log n )$-approximation algorithm. Later, in 2011 the article~\cite{cite:roy} settles the conjecture for the polygons that are \emph{weakly} visible from an edge and contain no holes by presenting a $6$-approximation algorithm for finding the minimum number of vertex guards. However, a counterexample for this result was presented later in 2019 by~\cite{cite:afk}. Ashur, Filtzer, and Katz~\cite{cite:afk} confirmed the conjecture
for the important case of weakly visible polygons, by presenting a $(2 + \epsilon)$-approximation algorithm for guarding such a polygon using vertex guards.

For the vertex guard and the point guard variants of the Art Gallery problem,~\cite{cite:2020} ruled out any $f(k)n^{o(k/\log k)}$ algorithm, where $k := |S|$ is the number of guards, for any computable function $f$, unless the exponential time hypothesis fails. These lower bounds almost match the $n^{O(k)}$ algorithms that exist for both versions of the problem.

An interesting result about the Art Gallery problem~\cite{cite:stoc}: The complexity status of this problem had not been resolved for many years. It has long been known that the problem is NP-hard, but no one has been able to show that it lies in $\NP$.
Recently, the computational geometry community became more aware of the complexity class $\exists \mathbb{R}$, which has been studied earlier by other communities. The class $\exists \mathbb{R}$ consists of problems that can be reduced in polynomial time to the problem of deciding whether a system of polynomial equations with integer coefficients and any number of real variables has a solution. It can be easily seen that $\NP \subseteq \exists \mathbb{R}$. The article~\cite{cite:stoc} proved that the Art Gallery problem is $\exists \mathbb{R}$-complete, implying that
(1) any system of polynomial equations over the real numbers can be encoded as an instance of the Art Gallery problem, and (2) the Art Gallery problem is not in the complexity class $\NP$ unless
$\NP = \exists \mathbb{R}$.

\section{Problem Definition and Our Result}
The point guarding problem where guards are restrict defines as the following.

\begin{pro}[\bf The Restricted Point Guarding problem]
\label{prob}
Given a simple polygon $\P$, the goal is to find the smallest set $\G$ of points of $\P$ so that each point of $\P$ is seen by at least one point of $\G$, and the points of $\G$ are constrained to be belong to the set of vertices of an arbitrarily dense grid. 
\end{pro}
We are going to convert the point guarding problem into a discretized version.
There is a constant factor lost from the discretization to the actual point guarding problem. The discretized version itself is an instance of a particular geometric set cover problem, where we will solve it with a constant-approximation factor. Denote a decomposition of $\P$ into a set of convex components by $\D$, and denote the numbers of cells of $\D$ by $|\D|$ which indicates the complexity of $\D$. Given a decomposition $\D$, the algorithm's running time proposed here is polynomial in $|\D|$. Considering a very dense grid (a given decomposition) with superpolynomial complexity, we can solve the general point guarding problem with a constant approximation factor. Otherwise, the proposed algorithm works for the restricted point guarding problem where the location of the guards is restricted to lie on a pre-determined grid.

Efrat and Har-Peled \cite{16} had also dealt with restricted point guarding problem. Bonnet and Miltzow improved their result. 
Assuming integer coordinates and a specific general position on the vertices of $\P$, Bonnet and Miltzow~\cite{cite:logopt} presented an $\mathcal{O}(\log|\OPT|)$-approximation algorithm for the point guarding problem.
Their result is one of the most recent works on the point guarding problem in 2020 \cite{cite:2020}. They use two assumptions:  

Assumption 1 (Integer Vertex Representation); vertices are given by integers, represented
in binary.

Assumption 2 (General Position Assumption); no three extensions meet in a point of $\P$
which is not a vertex and no three vertices are collinear. 

Bonnet and Miltzow combined a randomized pseudo-polynomial time algorithm proposed by Deshpande et al. in 2007 \cite{13} with another approach presented by Efrat and Har-Peled in 2006 \cite{16}.
The algorithm presented by Deshpande et al. is flawed and retracted by Bonnet and Miltzow. Efrat and Har-Peledattain presented a randomized polynomial-time $\mathit{O}(\log |\OPT_{grid}|)$ algorithms restricting guards to a very fine grid~\cite{16}, where $\OPT_{grid}$ denotes an optimal set of guards that is restricted to some grid. Efrat and Har-Peled could not prove that their grid solution approximates an optimal guard placement. By developing the ideas of Deshpande et al. in combination with the algorithm of Efrat and Har-Peled, Edouard Bonnet and Tillmann Miltzow proposed the first randomized polynomial-time approximation algorithm for simple polygons \cite{cite:logopt}. The running time of the algorithm used by Efrat and Har-Peled is polynomial in the complexity of $\P$, and the spread $\Delta$ of the vertices of the gallery, which $\Delta$ can be at most exponential in the input size. The spread of the vertices of the gallery defines like this: $\Delta = \frac{L}{\epsilon}$ where $L$ is the longest and $\epsilon$ is the shortest pairwise distances among the vertices of the polygon. Edouard Bonnet and Tillmann Miltzow used and modified the same grid; we know that the size of that grid could be superpolynomial. The grid width is polynomial in the inverse of the diameter of the polygon, which can be exponential in the description of the polygon. Considering the same assumptions, one can use the same grid in the first decomposition step of the algorithm proposed in this article and yield a constant factor for the point guarding problem. However, the running time of this approach would be polynomial in the size of the grid. 

As mentioned above after the discretization process, we have an instance of the geometric set cover problem. 

The set cover problem has a privileged place in computer science. In this problem, we are given a ground set $\cal X$ together with a family $\F$ = $\{Y_{1},Y_{2},...,Y_{n}\}$ where each $Y_{i}$ is a subset of $\cal X$. The family $\F$ covers $\cal X$; that is, $\cal X$ is contained in the union of the elements of $\F$. The goal is to find a subfamily $\cal G\subseteq \F$ with the minimum number of elements that also covers $\cal X$. The geometric set cover problem is the special case of the set cover problem in geometric settings, where $\cal X$ is a (usually finite) subset of some fixed dimensional Euclidean space such as the two-dimensional plane; in this paper, we will assume this to be the given simple polygon $\P$. The family $\F$ of subsets of $\cal X$ is induced by some family of objects, for example, disks, triangles, or visibility polygons. For instance, in the set cover problem with disks, we are given a set of disks, and we wish to find the minimum subset of disks that cover (whose union contains) $\cal X$.

This NP-hard problem admits several different polynomial-time algorithms that guarantee an approximation factor of $\mathcal{O}(\log |\cal X|)$. 
Given the constraints imposed by geometry, it is reasonable to expect that we can obtain approximation factors better than $\mathcal{O}(\log |\cal X|)$. 
We intend to apply an efficient approximation algorithm proposed by Kasturi Varadarajan in 2010 \cite{kasturi}. Kasturi used a randomized technique called Quasi-Uniform sampling. He managed to propose a sub-logarithmic approximation factor for fairly geometric objects. The approximation factor obtained by Kasturi in \cite{kasturi} is related to the well-known combinatorial problem of bounding the size of an $\epsilon$-net. Suppose that we have a collection of $n$ fairly general geometric objects (those of constant description complexity) such as disks and an integer $1 \leq L \leq n$. A $\frac{L}{n}$-net is a subset of the objects that covers all points in the plane that are $L$-deep, that is contained within at least $L$ of the objects. Nets of size better than $\mathcal{O}(\frac{L}{n}\log\frac{L}{n})$ are known to exist when the combinatorial complexity of the boundary of the union of the objects is near-linear~\cite{cite:epsilon13,cite:epsilon4,kasturi-2}.  
Kasturi~\cite{kasturi} proposed a randomized algorithm and used samples of nets for a strategy called Quasi-Uniform sampling. He used the Quasi-Uniform sampling approach to obtain an approximation factor of $2^{O(\log^{*}m)}\log h(m)$ for objects with union complexity $mh(m)$, where $m$ is the number of objects. The approach was modified and improved by Timothy M. Chan et al. in 2011 \cite{timothy-chan}. The improved result presented an $\mathcal{O}(1)$-approximation factor for the geometric set cover problem considering the assumption that the union complexity of the geometric objects should be near-linear in the number of the objects.

Here, we have a simple polygon $\P$ as the ground set, and we obtain a set of geometric objects whose union complexity is linear. The minimum number of these objects corresponds to a solution for the point guarding problem. Regardless of the structure and the complexity of initial decomposition ($\D$), we proposed an approach that provides a constant factor approximation for the point guarding problem.
To solve the instance of the geometric set cover problem, we use the Quasi-Uniform sampling approach presented by Kasturi Varadarajan \cite{kasturi}, and its modification that was improved by Timothy M. Chan et al.~\cite{timothy-chan}. This leads to a constant factor approximation for the point guarding problem. 

To discretize the point guarding problem into an instance of the geometric set cover problem, besides a given grid, we can use a decomposition approach proposed here. No matter how we decompose $\P$, each convex region obtained from the initial decomposition ($\D$) is called a spanning convex region, and the set of the spanning convex regions is denoted by $\SCR$. We subdivide every spanning convex region such that every region after this subdivision estimates a place for a guard. Every such region is called a \emph{guarding-region}. Every guarding-region carries a list of completely visible spanning convex regions. In fact, the characteristic of every guarding-region is that any spanning convex region that is completely visible to a point in a guarding-region $gr$ is also completely visible to every point of $gr$. The family of the visible lists of the guarding-regions is the family set of the geometric objects chosen for the set cover instance. 

Therefore, our overall approach can be summarized in the following steps:

-{\bf Step 1:} Choose a decomposition that decomposes $\P$ into a set of spanning convex regions.

-{\bf Step 2:} Compute approximate guard places inside the spanning convex regions (the guarding-regions), and obtain their corresponding visible lists.


-{\bf Step 3:} Formulate the geometric set cover problem and solve the discretized version with the approximation approach presented in~\cite{timothy-chan}.

As mentioned previously, we can choose any approach to decompose $\P$ into spanning convex regions, but the number of spanning convex regions determines the algorithm's running time. Regardless of the chosen strategy to compute spanning convex regions, we can compute the corresponding guarding-regions, and the approximation factor remains constant. Choosing a more dense grid that is fine enough for our goals leads to a more precise solution. In Section~\ref{sec:definitions}, we define the precise terminologies, and we will see how to compute guarding-regions. The suggested approach uses pure computational geometric techniques. Section~\ref{sec:scr} introduces a decomposing approach to compute spanning convex regions that takes $\mathcal{O}(n^{2^{k}})$ time to subdivide $\P$ into a list of spanning convex regions, where $k \leq \log n$ is a given positive integer. 
Even for a tiny amount of $k$, the final grid is remarkably dense, and 
as $k$ gets larger, the final grid gets exponentially denser.
Section~\ref{analyze} illustrates an analysis for Algorithm~\ref{algo.main}, and finally Section~\ref{sec:discussion} covers conclusion remarks.

\section{The Discretization Technique}
\label{sec:definitions}

 Consider a sub-region $\alpha$ of $\P$. A guard $g$ may situate in either of the following positions:   1- The guard $g$ does not see any point of $\alpha$. 2- The guard $g$ can see $\alpha$ partially. 3- The guard $g$ covers $\alpha$ entirely.
In the first and second situations, the proposed algorithm considers $\alpha$ as invisible for $g$. That is, for $\alpha$ to be counted as \emph{area-visible} to $g$, all of its points must be visible to $g$. A sub-region $sr$ is area-visible to another sub-region or segment denoted as the $source$, if all points of $sr$ are visible to the $source$. 

\begin{definition}
\emph{Area-Visibility}: A region $r$ is called area-visible to a source (a point, a segment, or a region), if all the points of $r$ are visible to all the points of the $source$.
\end{definition}

\begin{definition}
\label{def:sc-region}
\emph{SC-Regions}: Each minimal region formed from a decomposition of the given polygon $\P$ into a set of convex components is called a sc-region, short form a spanning convex region. The set containing all of the spanning convex regions is denoted by $\SCR$.
\end{definition}
Figure~\ref{fig.convex-regions} (a) illustrates an example of decomposing $\P$ into sc-regions.
 As mentioned before, $\P$ will be decomposed into other smaller special sub-regions called ``guarding-regions”. 
 If we put two guards in different arbitarary positions inside a guarding-region, the same sub-set of sc-regions is area-visible to both points. In order to compute guarding-regions we decompose every sc-region into guarding-regions. 
For a more precise definition for guarding-regions see the following definition. 
 \begin{definition}
 \label{def:guarding-region}
 \emph{guarding-regions}:
 Consider a simple polygon $\P$, decompose $\P$ into sc-regions. Given a query point $q$ in $\P$, a subset of sc-regions are area-visible to $q$. Call this subset as the \emph{visible list} of $q$ and denote this list by $\VL$. A guarding-region $gr$ is a convex polygon that the list of area-visible sc-regions to every point $p \in gr$ is the same.
  \end{definition}
We will estimate the visibility polygon of a guard in any arbitrary position inside a guarding-region by the visible list of that guarding-region. Denote the cardinality of a visible-list corresponding to a guarding-region $gr$ by $|\VL_{gr}|$.
Based on this definition, 
for a more straightforward presentation, we may refer to guarding-region's visibility list instead of a guard visibility polygon. 

In order to compute guarding-regions inside a sc-region, in the initial step we need to count on the area-visibility of each sc-region individually. A sub-region inside each sc-region that sees another sc-region completely is called a temporary sub-region, for short, we call each of them a temp-sub-region. See the following definition.
\begin{definition}
\emph{Temp-sub-region:} Given two sc-region $scr_{i}$ and $scr_{j}$ ($i\neq j$), suppose there is a sub-region inside $scr_{i}$ denoted by $tsr$ that $scr_{j}$ is area-visible to $tsr$, we call this sub-region a temp-sub-region.
\end{definition}
These temp-sub-regions will be computed by checking the edges of each sc-region individually. 
Each temp-sub-region eventually will be used to divide sc-regions into guarding-regions. 
Then, we sweep on each sc-region and decompose it to guarding-regions by the help of temp-sub-regions.
For the presentation to be simpler we define a visible list for each temp-sub-region too. Consider a temp-sub-region $tsr$ inside the sc-region $scr_{i}$ that $scr_{j}$ is area-visible to $tsr$ ($tsr_{i}(scr_{j})$), the visible list of $tsr$ contains $scr_{i}$ and $scr_{j}$. 

\subsection{Approximation Algorithm}
\label{full.discription}
Here we present an algorithm that computes the set of all the guarding-regions, this set is denoted by $\S$.
Algorithm \ref{algo.main} illustrates the pseudo-code of the approach. In Step 1, we decompose $\P$ into sc-regions, which $|\SCR|$ indicates the number of all the sc-regions (Recall that no matter how we compute sc-regions, we can compute the corresponding list of guarding-regions). Step 2 contains three $\for$ loops that corresponding to each sc-region, we check their edges one by one to find temp-sub-regions. Suppose $scr_{i}$ has $Ed_{i}$ edges. Note that any guard $g$ inside $scr_{i}$ may see other sc-regions ($scr_{j}$ $i\neq j$) only through $scr_{i}$'s edges. 
    Step 2 uses a procedure called $\Findtsr$ to compute temp-sub-regions correspond to every edge of a given sc-region. 
    
Step 3 receives all temp-sub-regions inside a given sc-region and decomposes them into guarding-regions. This step also uses a procedure called $Decompose$. This procedure receives a sc-region as a parameter. Then, by sweeping on that sc-region splits it into guarding-regions. The given sc-region denoted by $scr_{i}$ will be decomposed into a set of guarding-regions called $s_{i}$. The set $\S$ is the final set that contains all the guarding-regions, and let $|\S|$ denote the cardinality of $\S$. The algorithm computes every $s_{i}$ and adds it to $\S$ ({\small $1\leq i \leq {\tiny |\SCR|}$}). By the end of Algorithm~\ref{algo.main} $\S$ contains all guarding-regions, and their visible-lists. In fact, $\S$ contains pairs ($gr,\VL(gr)$) of all the guarding-regions and their corresponding visible lists.
A minimum subset of guarding-regions whose visible-list cover all the sc-regions solves the problem. 
In Step 4, by converting the problem to the geometric set cover problem we mange to solve the problem with a constant factor.  
\begin{theorem}
\label{theo.proof.opt}
There is an algorithm with $O(1)$-approximation factor for the Point Guarding problem and with $O(1)$-approximation factor with polynomial running time for the restricted Point Guarding problem (problem~\ref{prob}) if the size of the restricted grid is polynomial in the complexity of the given gallery ($\P$). 
\end{theorem}
Subsection~\ref{sec:approximation-analysis} deals with the approximation analysis of the algorithm. 
\begin{algorithm}[t]
\caption{Main Procedure}
\begin{algorithmic}[1]
\Procedure{Point Guarding($\P$)}{}
\State define $\longleftarrow$ = Add an element to a set/list.
\State define $tsr_{i}(scr_{j})$ = The temp-sub-region in the $i^{th}$ sc-region where $scr_{j}$ is area-visible to $tsr_{i}(scr_{j})$.
\State define //* = comment

//{\small {* \bf Step 1}}
\State Decomposing $\P$ into sc-regions and put them in $\SCR$.

//{\small {* \bf Step 2}}
\For{ i= 1 ; $i \leq  |\SCR|$; i++ } //{\small * i selects a sc-region}
    \For{ k=1 ; $k \leq  |Ed_{i}|$; k++ } //{\small * pick the $k^{th}$ edge of $scr_{i}$}
               \For{ j=1 ; $j \leq  |\SCR|$; j++ } //{\small * find every $tsr$ corresponding to every $scr_{j}$ inside $scr_{i}$ }
                  \State $tsr_{i}(scr_{j})$ = \Findtsr($ed_{ik}$,$scr_{j}$) //{\small * $ed_{ik}$ is the $k$th edge of $scr_{i}$ }
                  \State {\bf if} $tsr_{i}(scr_{j}) \neq \{\}$ {\bf do}
                  \State  $\ \ \ $ $\VL(tsr_{i}(scr_{j}))$ $\longleftarrow$ $scr_{j}$
               \EndFor
    \EndFor
\EndFor
//{\small {* \bf Step 3}}

//{\small \emph{*Extracting guarding-regions from temp-sub-regions in each $scr$:}}
\For{ i= 1 ; $i \leq  |\SCR|$; i++}
        \State Decompose($scr_{i}$)
        \State $s_{i}$ $\longleftarrow$ \emph{(the result regions, corresponding visible-lists)}.
      \For{ each~$gr$~in $s_{i}$ }
                 \State $\S$ $\longleftarrow$ $(gr,\VL_{gr})$;
            \EndFor
\EndFor

//{\small {* \bf Step 4}}

\State Formulate the set cover problem and solve it approximately.
\EndProcedure
\end{algorithmic}
\label{algo.main}
\end{algorithm} 
\begin{figure}[thp]
\begin{center}
\includegraphics[scale=0.7]{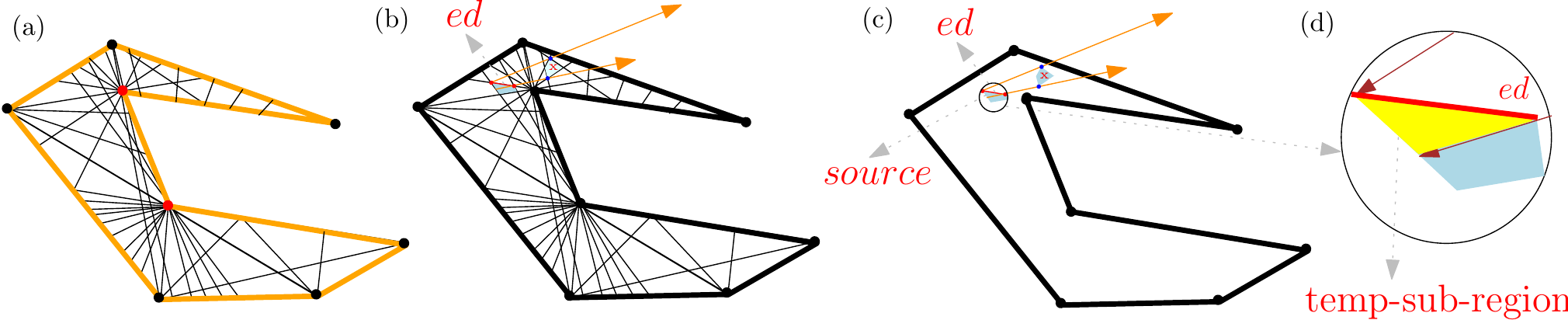}
\caption{ Figure (a) illustrates a given polygon that is decomposed into a set of sc-regions. Figure (b) shows a specific sc-region called source. Consider one edge of this sc-region denoted by $ed$, and a target sc-region denoted by $x$. Figure (c) illustrates $source$, $ed$, and $x$ more clearly. Figure (d) illustrates a temp-sub-region inside $source$ that every point of this temp-sub-region sees all the points of $x$.}
\label{fig.convex-regions}
\end{center}
\end{figure}  
\subsubsection{Procedure Findtsr(a segment, a sc-region)}
\label{subsec.algo.convex}

Consider a given segment $ed$ as a \emph{viewer}, we intend to find out if $ed$ can see a given spanning convex region called $target$ completely. 
Suppose that $ed$ is a given edge of a spanning convex region called $source$. 
    
We compute the complete-visibility-polygon of $ed$ ($\CVP(ed)$) and check if the target fits inside that. 
If the answer is yes, then we found a target that is area-visible to $ed$. 
We use the endpoints of $ed$ to find a temp-sub-region $tsr$ that can make the target area-visible. If so, we add it to \emph{the output set} of the $\Findtsr$ procedure. 
While we check the complete-visible-polygon of $ed$, consider the rays coming out of the endpoints of $ed$ that touch the boundary of the target. These rays must release from $ed$ and reach the target. The opposite direction from the target to $ed$ makes two half-lines that intersect $ed$ on its endpoints.

We extend these half-lines backward inside the source. The intersection between these two half-lines and the boundary $source$ creates a temp-sub-region $tsr$. By using binary search on the boundary of $source$ (which is a convex hull), we can compute this intersection in $\mathcal{O}(\log (|source|))$ time, where $|source|$ denotes the complexity of the source. Now we know that $target$ is area-visible to $tsr$.  

For a specific edge $ed$, to compute a temp-sub-region $tsr$; It takes $\mathcal{O}(n)$ time to compute its complete-visibility-polygon, and it takes $\mathcal{O}(\log |source|)$ to compute the boundary of $tsr$. Thus, it takes $\mathcal{O}(n+\mathcal{O}(\log |source|))$ time for each run of the $Findtsr$ procedure. Later we will see that the complexity if a spanning convex region is polynomial in terms of $n$, so the Find procedure takes $\mathcal{O}(n)$ time.

To compute all temp-sub-regions inside the source, we need to count on all the edges of the source. Note that there might only one point or an interval on $ed$ that sc-region is area-visible to that. 

\subsubsection{Procedure {\bf Decompose(spanning convex region)}}
\label{sub.algo.intersect-summary}
This procedure intends to decompose a given spanning convex region ($scr$) into guarding-regions.
According to Algorithm~\ref{algo.main} when we reach this procedure, every spanning convex region is divided into some temp-sub-regions that probably share common areas. Line 12 counts every spanning convex region, and line 13 calls the $Decompose$ procedure to obtain guarding-regions inside a given spanning convex region. The $Decompose$ procedure computes the guarding-regions inside $scr$ from the temp-sub-regions (see Figure~$\ref{fig.decopposing.cases}$). 
\begin{figure}[tp]
\begin{center}
\includegraphics[scale=0.5]{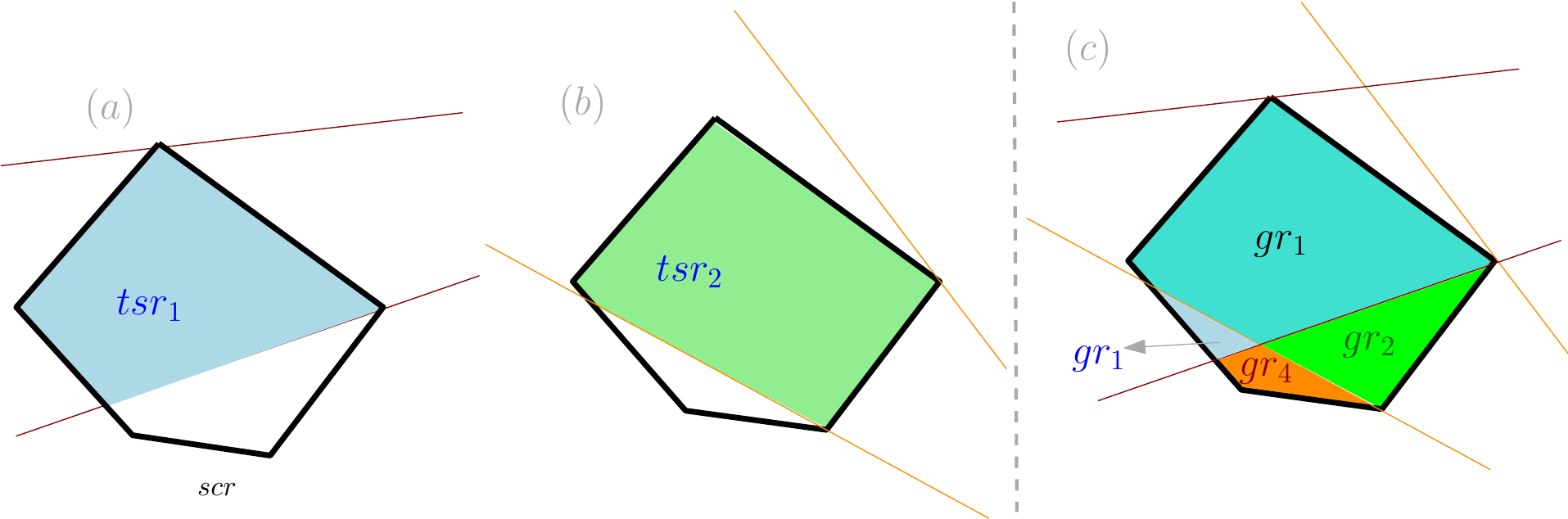}
\caption{ Three sc-regions are considered in this figure. The one that is illustrated and denoted by $scr$, and two others, let call them $scr_{1}$ and $scr_{2}$. Figure $(a)$ illustrates a temp-sub-region ($tsr_{1}$) in blue. Any guard inside $tsr_{1}$ can make $scr_{1}$ area-visible. Figure $(b)$ illustrates $tsr_{2}$, $scr_{2}$ is area-visible to $tsr_{2}$. Finally, figure~$(c)$ illustrates the decomposition of $scr$ into four guarding-regions. 
See figure $(c)$, $gr_{4}$ is illustrated in orange and makes only $scr$ area-visible. Any guard inside $gr_{2}$ can make $scr$ and $scr_{2}$ area-visible but not $scr_{1}$. }
\label{fig.decopposing.cases}
\end{center}
\end{figure}

A temp-sub-region $tsr$ is obtained from the area between two half-lines (backward rays). 
We aim to decompose every sc-regions $scr$ using their intersection with the area between these half-lines. 
Denote the number of temp-sub-regions inside $scr$, including $scr$, by $|scr|$, which indicates the complexity of $scr$.
For each temp-sub-region $tsr$ in $scr$, except for $scr$ itself, denote a starting half-line by $shl(tsr)$, and an ending half-line by $ehl(tsr)$. So, inside $scr$ a temp-sub-region $tsr$ is between two half-lines $shl(tsr)$ and $ehl(tsr)$. We are going to sweep on $scr$, regardless of the direction we move on, for a simpler presentation, we refer to these half-lines $shl(tsr)$ and $ehl(tsr)$.
From the previous steps of Algorithm~\ref{algo.main}, for every two starting and ending half-lines it is already specified what sc-regions ($\neq scr$) are visible to the points between any two half-lines $shl$ and $ehl$. 

\begin{lemma}
\label{proof:decompose}
A sc-region can be decomposed into guarding-regions in $\mathcal{O}(|scr|\log |scr|)$, where $|scr|$ denotes the complexity of a sc-region.
\end{lemma}
\begin{proof}
To decompose $scr$, we use a sweep-line denoted by $\cal{SL}$. We start sweeping parallel to an edge of $scr$. Various events may appear, $\cal{SL}$ checks every possibility and chooses its next action in each step. Appendix subsection~\ref{sub.algo.intersect} deals with the decomposition procedure. As the complexity of a sc-region is  $\mathcal{O}(|scr|)$ the sweeping approach may needs $\mathcal{O}(|scr|\log |scr|)$ time.
\end{proof}
\section{Computing Spanning Convex Regions}
\label{sec:scr}
In this section we propose a strategy to obtain spanning convex regions to be used in the initial step of Algorithm~\ref{algo.main}. 
 As mentioned before, we can use other decomposition strategies like the grid proposed by Bonnet and Miltzo~\cite{cite:logopt}. However, if we choose that grid in addition to the grid points, we need to include the reflex vertices\footnote{A reflex vertex is a vertex that the angle between the two edges of $\P$ incident on that vertex inside $\P$ is greater than $\pi$.}
 and intersections of extending lines (the extension of of two vertices is the line that contains those vertices). Otherwise, we run into an issue illustrated in Figure 1 of ~\cite{cite:logopt}. Here, we do not give details of that decomposition. Instead, we refer interested readers to \cite{cite:logopt} for further details.

First consider the following approach to decompose $\P$ into a set of spanning convex regions denoted by $\SCR$. Every iteration of the decomposition approach subdivides $\P$ into a set of cells. The final set of cells determines the set of spanning convex regions.

{\bf Decomposition Approach \#1}
\label{constructing-scr}
\begin{enumerate}
    \item 
Ignore the edges of $\P$. Connect every vertex of $\P$, and obtain every unique line passing from the vertices.
Compute all the intersections of these the lines.

\item Run this step $k$ times:
Connect all the intersection points and compute the lines crossing them. Compute the criss-cross of all the new lines with the previous lines.

\item Trace $\P$ and compute the intersection of this network with $\P$. 
\end{enumerate}


Consider a guard $g_{opt}$ in $\OPT$. A reflex vertex may block the visibility of $g_{opt}$ not to see a part of $\P$. Consider a situation where a blocking reflex vertex $\rf$ causes $g_{opt}$ to see a sc-region $scr$ partially. The rest of $scr$ must be covered by some other guards in $\OPT$. 
Without loss of generality, suppose that there is a visible area of $scr$ that is independently covered by $g_{opt}$. Denote \emph{the whole} visible part of $scr$ that is visible to $g_{opt}$ by $\Vp_{g_{opt},scr}$. If a decomposition approach can split $\P$ so that every visible part of each guard in the optimal solution suits sc-regions, then the solution of the restricted point guarding equals the solution of the general point guarding. This is because, in such a situation, every sc-region becomes \emph{completely} visible by at least one guard in the optimal solution. Recall that Algorithm~\ref{algo.main} counts on the complete visibility of the sc-regions to construct visible lists of the guarding-regions.

Consider one cell $c$ after step $i$ of the decomposition approach \#1. For two guards in the optimal solution $g^{1}_{opt}$ and $g^{2}_{opt}$, consider the visible parts $\Vp_{g^{1}_{opt},c}$ and $\Vp_{g^{2}_{opt},c}$. Suppose $\Vp_{g^{1}_{opt},c}$ and $\Vp_{g^{2}_{opt},c}$ share an area $A$ in $c$. The common area $A$ lies between two lines passing through $\overline{g^{1}_{opt}\rf_{1}}$ and $\overline{g^{2}_{opt}\rf_{2}}$, where $\rf_{1}$ and  $\rf_{2}$ are two blocking reflex vertices. After $k$ steps, if we divide $c$ so that the area $A$ gets covered by exclusive sc-regions, then the decomposition approach works for the general point guarding problem. Otherwise, we have to restrict guards to lie on the vertices of the grid constructed by the decomposition process. 

Consider the unique line passing through $\overline{g_{opt}\rf}$, for a guard in $\OPT$ and a blocking reflex vertex $\rf$. This line determines three sets of cells in each step of a decomposition approach: The set of cells that are completely visible to $g_{opt}$, the set of cells that are partially visible to $g_{opt}$, and the set of cells that are completely invisible to $g_{opt}$. As the polygon is a simple polygon, if a cell is partially visible or if it is invisible, it cannot become visible from elsewhere by the same guard. For $g_{opt}$, either a completely visible or a completely invisible cell in each step is called a \emph{Marked Cell by $g_{opt}$}.

Consider step $i$, we know that every cell $c$ is a convex hull. In the decomposition approach \#1, when we connect the intersection points to each other to obtain new lines for the subdivision operation in the $(i+1)^th$ step, we indeed connect the diagonals of every cell. So, every cell in the $i^{th}$ step gets at least divided to the sub-cells obtaining from the connection of its vertices (the diagonals). Except for a triangle, by connecting the diagonals of a convex hull, there will be at least one point in the middle of the convex hull that holds some sub-cells around itself. Denote such a middle point by $p \in c$. Consider the unique line passing through $\overline{g_{opt}\rf}$ that makes $c$ partially visible, and suppose that $\Vp_{g_{opt},c}$ is not completely visible by any other guard in $\OPT$. This line can only passes either through one side of $p$ or it can pass through $p$. No matter how this line may be situated, there is at least one sub-cell inside $c$ that after the subdivision process of the $(i+1)^{th}$ step becomes a marked cell for $g_{opt}$. So, for each partially visible cell in the $i^{th}$ step we know that after the $(i+1)^{th}$ subdivision step at least one sub-cell becomes marked for every guard in $\OPT$. So the area of partially visible cells gets smaller by each iteration. A simple analysis reveals that using a binary search after at most $\mathcal{O}(\log n)$ steps, the subdivided cells gets close enough to the shared area $A$ between the visible parts of the guards. Thus, if we set $k=\mathcal{O}(\log n)$ the decomposition approach works in the case of general point guarding problem.

However, if in the $i^{th}$ step, a cell $c$ is a triangle. Then, we have to split $c$ so that there exists a middle point inside $c$. This point can make the decomposition process split an unmarked $c$ into a set of sub-cells so that at least one sub-cell becomes a marked cell for each guard in $\OPT$ after the $(i+1)^{th}$ step. The decomposition approach \#2 presented in the following modifies the first decomposition approach to split the triangles correctly.  




{\bf Decomposition Approach \#2} (The main approach)
\label{constructing-scr}

\emph{$i^{th}$-order-Convex-Regions}: Each minimal region formed in the $i^{th}$ step of the subsequent approach is called an $i^{th}$-order-convex-region. 
Consider a small given constant denoted by $k$ that $i \leq k$.
\begin{enumerate}
    \item Ignore the edges of $\P$. Consider every unique line passing from connecting every vertex of $\P$, and compute
all the intersections of all the lines crossing each other. The set of cells obtained form this subdivision demonstrates the zero-order-convex-regions. 
\item In the step $i \geq 1$, run the following splitting process $k$ times. For each $i$-order-convex-region, draw all the unique lines crossing the specified segments designated in the below:
\begin{enumerate}
 \item Connect every two intersection points to each other.
    \item If the $(i-1)^{th}$-order-convex-region is a triangle connect the middle point of every edge to the third vertex of the triangle. 
\end{enumerate}
\item  Compute all the intersections of the above-mentioned lines.
\item  
Now trace $\P$ and compute the intersection of the boundary of $\P$ and the
above-mentioned lines.
\end{enumerate}
The $k^{th}$-order-convex-regions constructed from the above-mentioned approach specify the set of the spanning convex regions (see Figure~\ref{fig:decopposing.approach} (a) and (b)). 
\begin{figure}[tp]
\begin{center}
\includegraphics[scale=0.7]{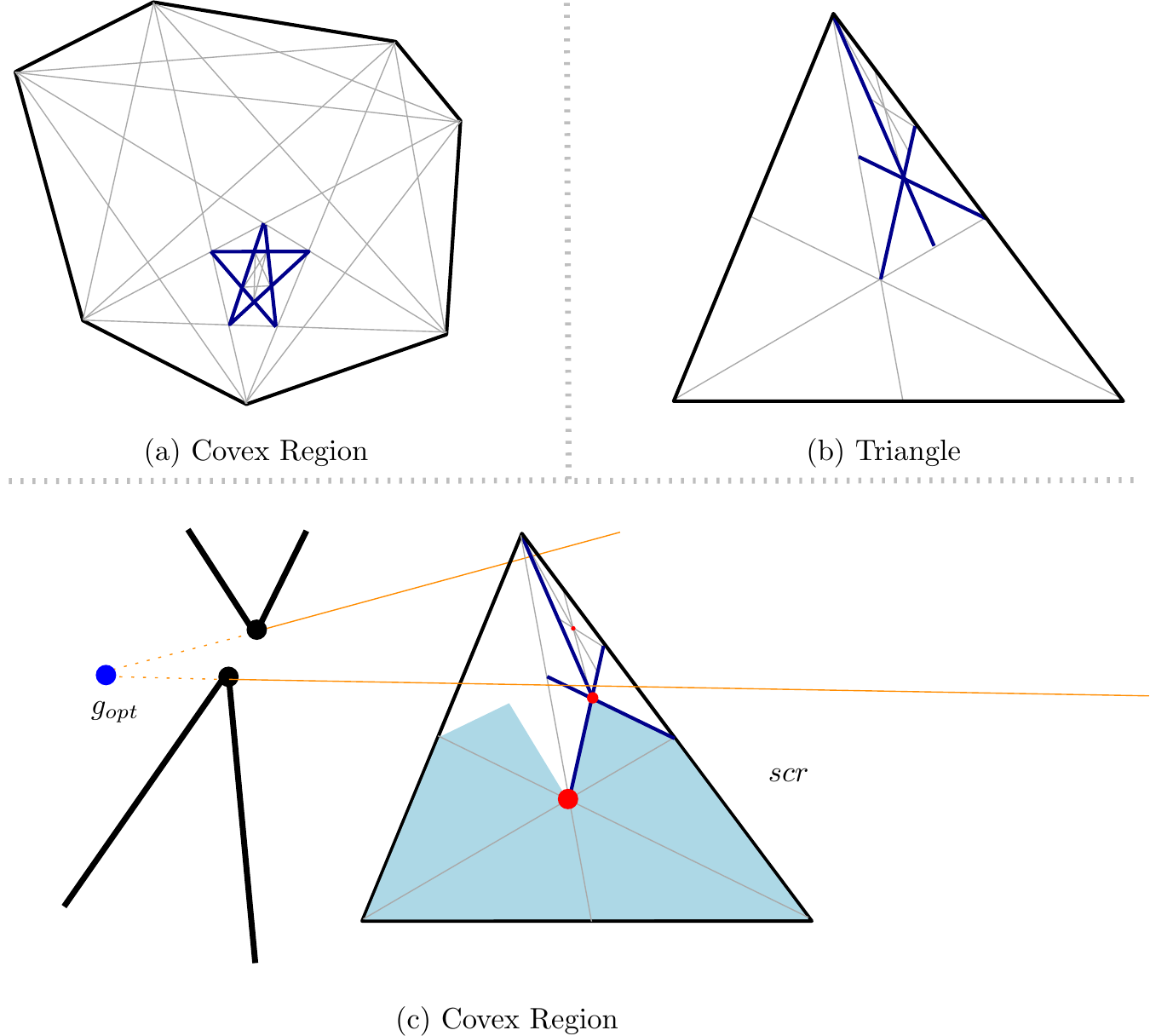}
\caption{Three steps of  decomposition process.}
\label{fig:decopposing.approach}
\end{center}
\end{figure}
   
\emph{Running Time Analysis of the Decomposition Approach:}

A simple analysis reveals that given a \emph{constant} $k$, the Decomposition Approach works in polynomial time. 


Without lost of generality we compute the complexity of Decomposition Approach \#1. We know that the complexity of the running time of Decomposition Approach \#2 is the same.

We know that $m$ lines can have at most $m^{2}$ intersections. The polygon has $n$ vertices. 
So, we have $\mathcal{O}(n^{2})$ lines, and $\mathcal{O}(n^{4})$ intersection-points. We connect these intersection-points to each other iteratively. For $i=2$ we have $\mathcal{O}(n^{8})$ cells. So the running time of the Decomposition Approach is $\mathcal{O}(n^{2^{k}})$. 

\section{Analysis of Algorithm~\ref{algo.main}}
\label{analyze}

This section covers analysis for Algorithm~\ref{algo.main}. First, note that due to line $6-8$ of the algorithm, we know that the guards chosen by the algorithm surely cover $\P$, so the algorithm provides a feasible solution. That is because when $i=j$, every spanning convex region will be added to get covered in the final solution. Although each spanning convex region $scr$ is decomposed into many temp-sub-regions, each of the result temp-sub-regions has $scr$ in its visible-list. So, all the spanning convex regions and consequently $\P$ will be covered.

\begin{lemma}
\label{lem:polynomial}
Algorithm~\ref{algo.main} takes polynomial time in the complexity of the grid size ($|\SCR|$).
\end{lemma}
\begin{proof}

A simple analysis reveals that the algorithm works in polynomial time if the number of the spanning convex regions, denoted by $|\SCR|$, be polynomial in complexity of the polygon. 
Every run of the Findtsr procedure takes $\mathcal{O}(n+\log|\SCR|)$ time, and every run of the Decompose procure takes $\mathcal{O}(n+\log|\SCR|)$ time complexity.

The Findtsr procedure will be invoked $\mathcal{O}(|\SCR|^{3})$ times, and the Decomposition procedure will be invoked $\mathcal{O}(|\SCR|)$ times. So, the time complexity of the three first steps is $\mathcal{O}(n|\SCR|^{3}+|\SCR|^{3}\log|\SCR| + |\SCR|^{2})$. Step $4$ of the algorithm uses the previous  polynomial time algorithms proposed by Kasturi Varadarajan~\cite{kasturi,kasturi-2} and Timothy M. Chan et al.~\cite{timothy-chan}.

\end{proof}

Suppose we are given a set $\M$ of $m$ points in $\P$, and a set $\S$ of visible-lists of guarding regions—the union of the visible-lists covers $\P$ and consequently $\M$. We seek the minimum subset of $\S$ that covers $\M$. This is clearly an instance of the combinatorial set cover problem. 
We intend to apply an efficient approximation algorithm proposed by Timothy M. Chan et al.~\cite{timothy-chan} in 2011, which is based on the approach proposed by Kasturi in 2010 \cite{kasturi}. 
To use Timothy M. Chan et al. and Kasturi's approach, we need to make sure that the complexity of the boundary of the union (union complexity) of the visible-lists is near-linear. 
The approximation factor obtained by rounding the natural linear programming relaxation of the set cover problem is related to the combinatorial problem of bounding the size of an $\epsilon$-net. For $m$ objects with union complexity $mh(m)$, Kastur obtained an approximation factor of $2^{O(\log^{*}n)}\log h(n)$ \cite{kasturi-2}, and Timothy M. Chan et al. improved the result into a constant factor of approximation. 

The decomposition we obtained yields a set of sc-regions that we need to cover all of them to cover $\P$. Pick an arbitrary point inside each sc-region and put them in $\M$. Since any visible-list of a guarding-region covers the entire region of a sc-region in the visible-list, choosing an arbitrary point from every sc-region is enough to cover $\P$, so $m=|\SCR|$. Furthermore, the union complexity of all the visible-lists equals the complexity of the number of sc-regions. Therefore, the union complexity of our objects is linear to the number of objects.


\section{Discussion}
\label{sec:discussion}
 The point guarding problem has known to be one of the preliminary versions of the Art Gallery problem. Since years ago, this version of the problem was interesting for many scientists and even regular people who wanted to guard a place. The problem has been modeled in many versions; however, the point guarding version that allows the guards to be placed inside a simple polygon was proved to be NP-hard. In 1998, Eidenbenz proved that there could not be any algorithm that solves the problem with a $(1+\epsilon)$ approximation factor. In other words, there is no PTAS for this problem. In 2020, which is more than twenty years after Eidenbenz's result, Bonnet and Miltzow presented a $\mathit{O}(\log n)$-approximation factor for the point guarding problem. Therefore, there was a gap between these results, and finding a constant approximation factor should be noteworthy.
 
 Here, we presented an algorithm with a $\mathcal{O}(1)$-approximation factor for the point guarding problem where guards are restricted to lie on an arbitrary dense grid.
 If we choose the grid proposed by Bonnet and Miltzow, we can solve the general point guarding problem with a constant factor approximation in 
 running time polynomial in the ratio of the longest and shortest pairwise distances of the vertices of the art gallery. 
 Otherwise, if we choose the approach proposed in this paper, the restricted point guarding problem can be solved in $\mathcal{\~O}(n^{2^{k}})$, where $k$ can be chosen to be fine enough. It remains open that if for a constant $k$ the grid solution is indeed an approximation of an optimal solution. That is, we do not know if there is constant $k$ so that the optimal solution restricted to the final grid is at most a constant times larger than the general optimal solution. 
 
For the sake of simplicity, we only considered the case of an art gallery without holes. For future work, we generalize the approach for the case of an art gallery with holes.




\enlargethispage{\baselineskip}



\small
\bibliographystyle{abbrv}

\appendix
\newpage
\section{Decomposing a spanning convex region}
\label{sub.algo.intersect}

In this subsection, we see how to decompose a given sc-region into distinct and disjoint guarding-regions. This subsection is the complete version of ~\autoref{sub.algo.intersect-summary}. Here, we also prove \autoref{proof:decompose}.

Form the previous steps of Algorithm~\ref{algo.main} we know that a temp-sub-region is created based on the collision of two half-lines with the boundary of the given sc-region. The approach sweeps on the given sc-region. When the sweeper reaches a half-line that starts/ends a temp-sub-region, we add/eliminate the visible sc-region corresponding to that temp-sub-region to the visible-list of currently in progress guarding-region. If the sweeper reaches an intersection point of many half-lines, there could be more than one sc-region that we should add or eliminate from the currently in progress guarding-regions. The approach is presented in details in the following:

\emph{Decomposition-Process:}
Consider a sc-region $scr_{i}$. To decompose $scr_{i}$, we use a sweep-line called $\cal{SL}$. The sweeping process starts from an edge $e$ of $scr_{i}$. The sweep-line $\cal{SL}$ first contains $e$, and moves on $scr_{i}$ parallel to $e$ till it reaches the end of $scr_{i}$ where $\cal{SL}$ sees all the points of $scr_{i}$. 

The final decomposition of $scr_{i}$ results in several newly created guarding-regions. The Decompose procedure puts the final guarding-regions in a set called $s_{i}$. 
Note that we have to compute a visible-list for every new guarding-region $gr$. We already know that $\VL(gr \in scr_{i})$ contains $scr_{i}$. 
The sweep-line $\cal{SL}$ has a list of currently in progress guarding-regions that should decide where to end them. Denote this list by $\inPL$. Denote the number of guarding-regions in the current list of the sweep-line by $|\inPL|$. 
During the sweeping process $\cal{SL}$ may encounter with three cases: 1) It may reach a starting half-line $shl$. 2) It may reach an intersection-point of several $shl$ or $ehl$ half-lines. 3) Or it may reach an ending half-line $ehl$.

\begin{figure}[tp]
\begin{center}
\includegraphics[scale=0.6]{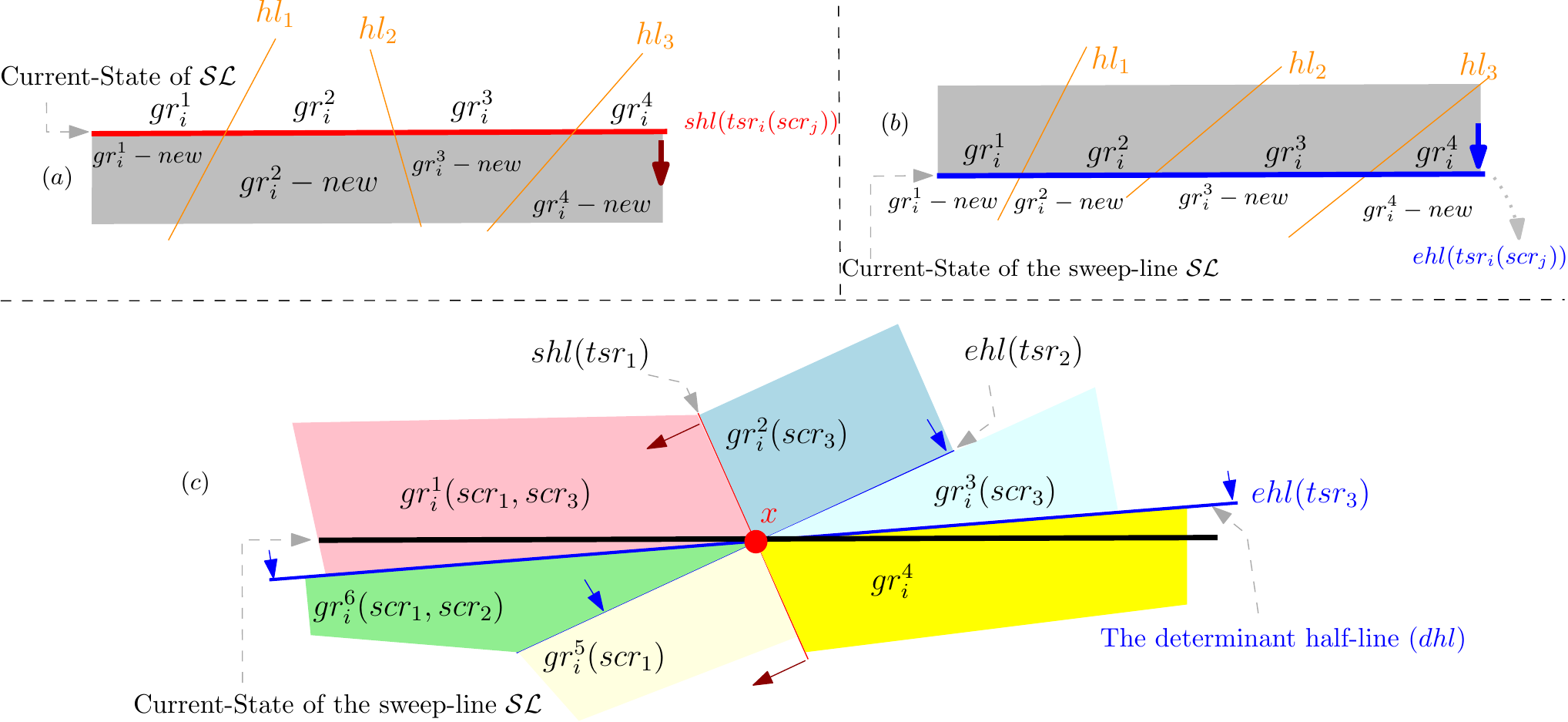}
\caption{
Figure $(a)$ shows that when $\cal{SL}$ reaches a $shl(tsr_{i}(scr_{j}))$ half-line we should create new guarding-regions with new visible-lists. We have to add $scr_{j}$ to the visible-list of the guarding-regions that $\cal{SL}$ has just passed. Figure $(b)$ illustrates a similar situation where $\cal{SL}$ reaches an ending half-line $ehl(tsr_{i}(scr_{j}))$. In such a situation we have to end the currently in process guarding-regions and create new guarding-regions with new corresponding visible-lists that do not contain $scr_{j}$ in their visible-lists. In Figure $(c)$ $\cal{SL}$ may reaches an intersection point $x$ of several half-lines (either $shl$ or $ehl$). In this case, there is a determinative half-line $dhl$ that passes $x$. This half-line can be either a starting half-line or an ending half-line. Since the sweep-line $\cal{SL}$ has encountered $dhl$ previously, the half-line $dhl$ do not make any new guarding-region itself. Nevertheless, the intersection-point $x$ make new guarding-regions from previously visited half-lines. See Figure (c), $gr^6$ exhibited in green or $gr^5$ in white are the new guarding-regions that are starting from $x$. However, $gr^2$ and $gr^3$ are ended at $x$. For $gr^1$ and $gr^4$ it is sufficient to update their visible-lists. Based on the type of each half-line we can easily create the new guarding-regions around $x$, and also we can determine their corresponding visible-lists. Note that if $dhl(tsr_{i}(scr_{j}))$ is parallel to the direction of $\cal{SL}$ then the sweep-line reaches a new half-line and a point, and the creation process of the new guarding-regions is a combined procedure of the above-mentioned strategies. In such a case, we have to consider $scr_{j}$ as an additional or removal region in the visible-lists of the new guarding-regions below $dhl(tsr_{i}(scr_{j}))$.}
\label{fig.decopposing.cases}
\end{center}
\end{figure}
Consider an arbitrary $scr_{k} \neq scr_{i}$ and suppose $tsr \in scr_{i}$ can make $scr_{k}$ area-visible. 
There are three different cases. See \autoref{fig.decopposing.cases} as an example.
In this figure, the numbers on power of $gr$ are normal indicators to itemize different guarding-regions. This figure illustrates three different cases where a sweep-line may encounter while tracing a sc-region $(scr_{i})$ to decompose it into guarding-regions. The three cases:
 
 1,2) If $\cal{SL}$ reaches a starting half-line $shl(tsr_{i}(scr_{k}))$ or an ending half-line $ehl(tsr_{i}(scr_{k}))$;
    Do these steps:
    \begin{enumerate}
        \item  End all the guarding-regions in the $\inPL$.
        That is because all the previously submitted guarding-regions in $\inPL$ are ended by the the new line. 
        \item  Add every newly ended guarding-region $gr$ to $s_{i}$ (the output set). 
        \item Based on the number of crossing lines with the half-line $shl(tsr_{i})/ehl(tsr_{i})$, create new guarding-regions in $\inPL$. Each of such new guarding-regions starts from $e/shl(tsr)$ and has its own boundary.
        \item Create a visible-list for each new $gr \in \inPL$ ($\VL(gr)$) with the initiation value of $scr_{i}$.
       \item For each $gr \in \inPL$, if $gr_{x}$ has $gr_{y}$ on the other side (the already swept side) of $e/shl(tsr)$, set $\VL(gr_{x}) \longleftarrow \VL(gr_{y})$.
       Note that we only check one line at a moment, so only one temp-sub-region and consequently the visibility of one sc-region is checked, and other visible sc-regions are still visible.
       \item For each new $gr \in \inPL$,
       
       {\bf if} $\cal SL$ reaches $shl(tsr_{i}(scr_{k}))$;

       $\ \ \ $Add $scr_{k}$ to $\VL(gr)$.
       
       {\bf else} $\cal SL$ reaches $ehl(tsr_{i}(scr_{k}))$;
       
        $\ \ \ $Remove $scr_{k}$ from $\VL(gr)$.
       
    \end{enumerate}
  3)
   If $\cal{SL}$ reaches an intersection-point $x$ of $k$ $shl$ half-lines and $k'$ $ehl$ half-lines;\\ Without loss of generality suppose all of these half-lines are on distinct lines, and assume that the sweep-line is slightly tilted so that it considers the intersections one by one. 
   We already know that $|inPL|=\frac{k+k'}{2}$ (see \autoref{fig.decopposing.cases}(c)). So, all the half-lines are already seen by $\cal{SL}$. The most recently visited half-line is either a $shl$ or a $ehl$. Again, w. l. o. g. suppose it is $shl(tsr_{i}(scr_{k}))$, and $shl(tsr)$ is parallel to the direction of the sweep-line $\cal{SL}$.
   
   Do the following steps:
        \begin{enumerate}
        \item  End every guarding-region in $\inPL$. These guarding-regions are above $shl(tsr)$, where $\cal{SL}$ was moving before it reaches $x$.
        \item  Add every newly ended guarding-region to $s_{i}$.
        \item  Based on the lines crossing $x$, create new guarding-regions in $\inPL$.
        \item Create a visible-list for each new guarding-region ($gr$) in $\inPL$ ($\VL(gr)$).
        \item Add $scr_{k}$ to every $\VL(gr)$. 
        \item There are nv  newly guarding-regions created in $s_{i}$. For each $gr$ check it with each half-line $hl$ intersected in $x$.
        
        (a) If $hl$ is a starting half-line $shl(tsr_{i}(scr_{k}))$ and $gr$ lies between $shl(tsr_{i}(scr_{k}))$ and $ehl(tsr_{i}(scr_{k}))$, then add $scr_{k}$ to $\VL(gr)$.
        
        (b) If $hl$ is an ending half-line $ehl(tsr_{i}(scr_{k}))$ and $gr$ is not between $shl(tsr_{i}(scr_{k}))$ and $ehl(tsr_{i}(scr_{k}))$ anymore, then remove $scr_{k}$ from $\VL(gr)$. 
        
        \end{enumerate}
Set $s_{i}$ from the above-mentioned approach as the output, which is the decomposed set of a given sc-region.

\section{List of Notations}
For a better presentation and the convince of the reader, we list the frequent notations that we used throughout the paper.

$scr$ = a sc-region

$gr$ = a guarding-region

$\rf$ = a reflex-vertex


$\VP$ = visibility polygon of a point or a guard

$\Vp$ = visible part of a guard inside a sc-region

$\VL$ = visible-list of a guarding-region

$\S$ = the set of all guarding-regions after the decomposition of $\P$ into guarding-regions by the beginning of Step 4 of Algorithm~\ref{algo.main}

$\OPT$ = the set of guards in the optimal solution of the point guarding problem

$\OPT(\S)$ = the set of guards in the optimal solution that can be obtained from $\S$.

\end{document}